\documentclass[11pt,a4paper]{article}
\usepackage{amsmath,amssymb,amsthm,graphicx,braket,multirow,authblk,amsthm,dcolumn,latexsym,subcaption,mathtools,cite,cancel}
\usepackage[normalem]{ulem}
\usepackage[a4paper]{geometry}
\DeclareSymbolFontAlphabet{\amsmathbb}{AMSb}%
\usepackage[most]{tcolorbox}
\usepackage[T1]{fontenc}
\usepackage[cp1250]{inputenc}
\usepackage{graphicx}
\usepackage[unicode=true,
bookmarks=true,bookmarksopen=false,
breaklinks=false,pdfborder={0 0 0},colorlinks=true]
{hyperref}
\usepackage{xcolor}
\definecolor{cblue}{rgb}{0.16, 0.32, 0.75}
\definecolor{cred}{rgb}{0.7, 0.11, 0.11}
\hypersetup{%
	,linkcolor=cred
	,citecolor=cblue
	,urlcolor=black
}
\def\<{\langle}
\def\>{\rangle}
\newtheorem{Example}{Example}
\newtheorem{theorem}{Theorem}
\newtheorem{proposition}{Proposition}
\newtheorem{corollary}{Corollary}
\newtheorem{definition}{Definition}
\newtheorem{remark}{Remark}

\def\oper{{\mathchoice{\rm 1\mskip-4mu l}{\rm 1\mskip-4mu l}
		{\rm 1\mskip-4.5mu l}{\rm 1\mskip-5mu l}}}
\newcommand{\rT}{{\rm T}}

\DeclareMathAlphabet\mathbfcal{OMS}{cmsy}{b}{n}

\begin{document}	

	
	
	
\title{\textbf{A class of Schwarz qubit maps with diagonal unitary and orthogonal symmetries}}


	\author[$\hspace{0cm}$]{Dariusz Chru\'sci\'nski\footnote{darch@fizyka.umk.pl} and Bihalan Bhattacharya\footnote{bihalan.bhattacharya@v.umk.pl} }
	\affil{\small Institute of Physics, Faculty of Physics, Astronomy and Informatics, Nicolaus Copernicus University, Grudziadzka 5/7, 87-100 Toru\'n, Poland}

	\maketitle
	\vspace{-0.5cm}	
	
	\begin{abstract}
A class of unital qubit maps displaying diagonal unitary and orthogonal symmetries is analyzed. Such maps already found a lot applications in quantum information theory. We provide a complete characterization of this class of maps showing intricate relation between positivity, operator Schwarz inequality, and complete positivity. Finally, it is shown how to generalize the entire picture beyond unital case (so called generalized Schwarz maps). Interestingly, the first example of  Schwarz but not completely positive map found by Choi belongs to our class. As a case study we provide a full characterization of Pauli maps. Our analysis leads to generalization of seminal Fujiwara-Algoet conditions for Pauli quantum channels.

	\end{abstract}
	
	\maketitle
	
\section{Introduction}\label{sec:1}

Symmetry plays an important role in modern physics. Very often the presence of symmetry
enables one to simplify the analysis of the
corresponding problems and leads to much deeper understanding and the most elegant mathematical formulation
of the corresponding physical theory. In quantum information theory \cite{QIT} the idea of symmetry is realized in a natural way on the level of maps (e.g. quantum channels) and quantum states. Consider a unitary $n$-dimensional representations $U$ of a symmetry group $G$.  A  linear map $\Phi : \mathcal{M}_n \to \mathcal{M}_n$ is $U$-covariant if

\begin{equation}\label{}
  U \Phi(X) U^\dagger = \Phi(UXU^\dagger) ,
\end{equation}
for all $X \in \mathcal{M}_n$ and all elements $U$ from the group (we identify the group with its unitary representation). $\Phi$ is called conjugate $U$-covariant if

\begin{equation}\label{}
  U \Phi(X) U^\dagger = \Phi(\overline{U}XU^\rT) ,
\end{equation}
where $U^\rT = \overline{U}^\dagger$ denotes the transposition. The studies of covariant completely positive maps where initiated by Scutaru \cite{Scutaru} and then developed by Holevo \cite{Holevo-1,Holevo-2} (for more recent analysis cf. e.g. \cite{C1,C2,C3,C4}). On the level of quantum states a bipartite state $\rho$ living in $\mathbb{C}^n \otimes \mathbb{C}^n$ is $U\otimes U$-invariant if

\begin{equation}\label{W}
  U \otimes U \rho U^\dagger \otimes U^\dagger = \rho ,
\end{equation}
and it is $U\otimes \overline{U}$-invariant if

\begin{equation}\label{I}
  U \otimes \overline{U} \rho U^\dagger \otimes U^{\rm T} = \rho .
\end{equation}
In the case when $G=\mathcal{U}(n)$ the $U\otimes U$- and  $U\otimes \overline{U}$-invariant states are nothing but celebrated Werner and isotropic states, respectively \cite{Werner,Iso,HHHH}. Both Werner and isotropic states play a prominent role in quantum information theory \cite{HHHH}. Note that $\Phi$ is $(U,U)$-covariant iff the corresponding Choi matrix

\begin{equation}\label{}
  C_\Phi = \sum_{i,j=1}^n |i\>\<j| \otimes \Phi(|i\>\<j|) ,
\end{equation}
is $U\otimes \overline{U}$-invariant. Similarly, $\Phi$ is $(U,\overline{U})$-covariant iff $C_\Phi$ is $U\otimes U$-invariant.

An important classes of covariant maps and invariant states correspond to $G=\mathcal{D}\mathcal{U}(d)$, i.e. diagonal unitary $n\times n$ matrices. This class of so called diagonal unitary covariant maps was recently analyzed in \cite{Ion1,Ion2,Ion3,Ion4}. In particular it was shown \cite{Ion3} that diagonal covariant maps satisfy the celebrated PPT$^2$ conjecture \cite{PPT1,PPT2}, that is, composition of any two PPT maps is entanglement breaking.

In this paper we analyze diagonal unitary covariant qubit maps. Such maps are called phase-covariant

\begin{equation}\label{I}
  \Phi(e^{i \varphi \sigma_z} X e^{-i \varphi \sigma_z}) = e^{i \varphi \sigma_z}\Phi(X) e^{-i \varphi \sigma_z} ,
\end{equation}
or conjugate phase-covariant

\begin{equation}\label{II}
  \Phi(e^{i \varphi \sigma_z} X e^{-i \varphi \sigma_z}) = e^{-i \varphi \sigma_z}\Phi(X) e^{i \varphi \sigma_z} ,
\end{equation}
for arbitrary real phase $\varphi$. The qubit maps satisfies both (\ref{I}) and (\ref{II}) iff

\begin{equation}\label{III}
  \Phi(\sigma_z X \sigma_z) =  \sigma_z \Phi(X)  \sigma_z ,
\end{equation}
that is, $\Phi$ commutes with a unitary Pauli map $\sigma_z \bullet \sigma_z$. In this case $G$ consists of diagonal orthogonal $2 \times 2$ matrices. In particular  we analyze the operator Schwarz inequality

\begin{equation}\label{S!}
  \Phi(X^\dagger X) \geq \Phi(X^\dagger)\Phi(X) .
\end{equation}
A unital map $\Phi : \mathcal{M}_n \to \mathcal{M}_m$ is called a {\em Schwarz map} if (\ref{S!}) holds for all $X \in \mathcal{M}_n$. It is well known that Schwarz property is stronger than positivity and weaker than complete positivity \cite{Paulsen,Stormer,BHATIA,WOLF,PR}. The first example of a Schwarz map $\Phi : \mathcal{M}_2 \to \mathcal{M}_2$ which is not completely positive was provided by Choi \cite{Choi2}. Interestingly, unital positive maps satisfy a weaker condition, i.e. (\ref{S!}) for all $X=X^\dagger$, called a Kadison inequality \cite{Kadison}.

Completely positive maps provide mathematical representation of quantum channels and more generally quantum operations. However, it is just a Schwarz property which allows to prove several monotonicity results in quantum information theory \cite{Petz,Alex,Carlen}. Interestingly, Schwarz property  is sufficient for the proper description of the asymptotics of  open quantum systems \cite{Amato,Amato2}.  Recently, so called Schwarz-divisibility was introduced for the analysis of Markovianity of open quantum dynamics \cite{KS} and in a recent paper \cite{GF} it was shown that relaxation rates of Markovian semigroups of qubit unital Schwarz maps satisfy the universal constraint which modify the corresponding constraint for completely positive semigroups \cite{PRL,LAA}.

The paper is organized as follows: {Section \ref{Se-II}} provides a brief introduction to the class of covariant qubit maps we are going to analyze. Moreover, positivity and complete positivity within this class is discussed. Condition for complete positivity is very simple. However, it seems that positivity was not analyzed in full generality within this class of maps.  In Section \ref{S-III} we characterize Schwarz unital maps which are phase-covariant or conjugate phase-covariant. Section \ref{S-IV} generalizes this result for the whole class defined by (\ref{III}). Interestingly, the entire analysis has a simple geometric interpretation. As an instructive illustration we provide a full characterization of well known Pauli maps in Section \ref{S-Pauli}. In this case conditions for positivity and complete positivity are well known. However, full characterization of Schwarz property was missing in the literature. Again, the entire characterization gives rise to suggestive geometric picture. In section \ref{S-IVa} we extend the analysis to include so called generalized Schwarz maps \cite{Alex}. Final conclusions are collected in Section \ref{S-C}. Additional technical details are presented in the Appendix.

\section{A class of covariant qubit maps}   \label{Se-II}

The qubit maps satisfying diagonal symmetries (\ref{I})-(\ref{III}) are fully characterized by the following


\begin{proposition} A qubit map satisfies (\ref{III}) iff



\begin{equation}\label{!}
  \Phi(X) = \sum_{i,j=1}^2 a_{ij} |i\>\<j| X |j\>\<i| + \Big( \lambda P_1 X P_2 + \overline{\lambda} P_2 X P_1 \Big) + \Big( \mu P_1 X^{\rm T}  P_2 + \overline{\mu} P_2 X^{\rm T}  P_1 \Big) ,
\end{equation}
where $P_i = |i\>\<i|$. Moreover, $\Phi$ is phase-covariant, i.e. satisfies (\ref{I}) iff $\mu=0$, and it is conjugate phase-covariant  i.e. satisfies (\ref{II}) iff $\lambda=0$.
\end{proposition}
The proof is straightforward. In the matrix form the map (\ref{!}) has the following structure

\begin{equation}\label{}
  \Phi(X) = \left( \begin{array}{cc} a_{11} X_{11} + a_{12} X_{22} & \lambda X_{12} + \overline{\mu} X_{21} \\ \overline{\lambda} X_{21} + {\mu} X_{12} & a_{21} X_{11} + a_{22} X_{22} \end{array} \right) ,
\end{equation}
where $X = \sum_{i,j} X_{ij} |i\>\<j|$. It is clear that $\Phi(X^\dagger) = \Phi(X)^\dagger$, i.e. $\Phi$ is Hermiticity-preserving if $a_{ij} \in \mathbb{R}$.  Moreover, $\Phi$ is trace-preserving if $\sum_i a_{ij}=1$ and unital if $\sum_j a_{ij}=1$.

One easily finds the corresponding Choi matrix

\begin{equation}\label{CHOI}
  C = \left( \begin{array}{cc|cc} a_{11} & . & . & \lambda \\  . & a_{21} & \overline{\mu} & .  \\ \hline   . & \mu & a_{12} & . \\  \overline{\lambda} & . & . & a_{22} \end{array} \right) ,
\end{equation}
where to make the matrix structure more transparent we replaced zeros by dots.

\begin{corollary} $\Phi$ is completely positive if and only if $a_{ij} \geq 0$ and

\begin{equation}\label{CP}
  |\lambda| \leq \sqrt{a_{11} a_{22}}  , \ \ \  |\mu| \leq \sqrt{a_{12} a_{21}} .
\end{equation}

\end{corollary}

This class of maps contains many well known examples of qubit quantum channels positive maps important in entanglement theory.

\begin{Example}[Qubit quantum channels]

A prominent examples of a phase-covariant map is an amplitude damping channel

\begin{equation}\label{}
  \Phi_\eta(\rho) := K_1 \rho K_1^\dagger + K_2 \rho K_2^\dagger ,
\end{equation}
with the following Kraus operators

\begin{equation}\label{}
  K_1 = \left( \begin{array}{cc} 1 & 0 \\ 0 & \sqrt{\eta} \end{array} \right) \ , \ \ \
   K_2 = \left( \begin{array}{cc} 0 & \sqrt{1-\eta} \\ 0 &  0 \end{array} \right) \ , \ \ \ \eta \in[0,1] .
\end{equation}
It corresponds to (\ref{!}) with

\begin{equation}\label{}
  a_{11}=1 \ , \ \ a_{12} = 1- \eta \ , \ \ a_{21} = 0 \ ,  \ \ a_{22} = \eta \ , \ \lambda= \sqrt{\eta}\ , \ \ \mu=0 .
\end{equation}
A generalized amplitude damping channel provides an example of a map with diagonal orthogonal symmetry

\begin{equation}\label{}
  \Phi^{(p)}_{\eta}(\rho) := \sum_{i=1}^4 K_i \rho K_i^\dagger , \ \ \ \ p,\eta \in [0,1]
\end{equation}
where

\begin{equation}\label{}
  K_1 = \sqrt{p}\left( \begin{array}{cc} 1 & 0 \\ 0 & \sqrt{\eta} \end{array} \right) \ , \ \ \
   K_2 = \sqrt{p}\left( \begin{array}{cc} 0 & 0  \\  0 &  \sqrt{1-\eta} \end{array} \right) \ ,
\end{equation}
and
\begin{equation}
    K_3 = \sqrt{1-p}\left( \begin{array}{cc} \sqrt{\eta}  & 0 \\ 0 & 1\end{array} \right) \ , \ \ \
   K_4 = \sqrt{1-p}\left( \begin{array}{cc} 0 & 0  \\ \sqrt{1-\eta}  & 0 \end{array} \right) \ , \ \ \
\end{equation}
It corresponds to (\ref{!}) with

\begin{equation*}\label{}
  a_{11}= p + (1-p) \eta \ , \ \ a_{12} = p(1- \eta) \ , \ \ a_{21} = 1-a_{11} \ , \ \ 
  a_{22} = 1 - a_{12} \ , \ \ 
  \lambda= (1-p)\,\sqrt{\eta}\ , \ \ \mu=p\, \sqrt{1-\eta}  .
\end{equation*}
Another well known example of a qubit map with diagonal orthogonal symmetry is realized by a Pauli channel

\begin{equation}\label{Pauli-M}
  \Phi(X) = \sum_{\alpha=0}^3 p_\alpha \sigma_\alpha X \sigma_\alpha ,
\end{equation}
where $p_\alpha \geq 0$ and $\sum_\alpha p_\alpha = 1$. It corresponds to (\ref{!}) with

$$   a_{11}=a_{22} = p_0+p_3\ , \ \ a_{12}=a_{21} = p_1+p_2 \ , \ \ \lambda = p_0-p_3 \ , \ \ \mu = p_1 - p_2 .   $$
This class of maps is further analyzed in Section \ref{S-Pauli}. This class contains quantum channels important for error correcting codes \cite{QIT,Wilde}: a bit flip channel ($p_0 = p$ and $p_1 = 1-p$), phase-flip channel ($p_0 = p$ and $p_3 = 1-p$), and bit-phase flip channel ($p_0 = p$ and $p_2 = p$).

\end{Example}

\begin{proposition} \label{Pro-P} $\Phi$ is positive if and only if $a_{ij} \geq 0$ and

\begin{equation}\label{P}
  |\lambda| + |\mu| \leq \sqrt{a_{11} a_{22}} +  \sqrt{a_{12} a_{21}} .
\end{equation}
\end{proposition}
For the proof see an Appendix. It is clear that (\ref{CP}) implies (\ref{P}) but the converse is obviously not true.
If the map (\ref{!}) is unital to simplify notation let us use the following convention

\begin{equation}\label{}
   \left( \begin{array}{cc} a_{11} & a_{12} \\ a_{21} & a_{22}  \end{array} \right) =
    \left( \begin{array}{cc} a & 1-a \\ 1-b & b  \end{array} \right) \ , \ \ \ \ a,b\in[0,1] .
\end{equation}
In this case condition (\ref{P}) reduces to

\begin{equation}\label{P1}
  |\lambda| + |\mu| \leq \sqrt{a b} +  \sqrt{(1-a)(1-b)} .
\end{equation}
Note, that $\Phi$ is bi-stochastic, i.e. both unital and trace-preserving, iff $a=b$. In this case the above condition simplifies to

\begin{equation}\label{P2}
  |\lambda| + |\mu| \leq 1 .
\end{equation}




\section{Schwarz maps with diagonal unitary symmetry}   \label{S-III}

Consider a unital phase-covariant map (\ref{!}) corresponding to $\mu=0$ and  conjugate phase-covariant  map corresponding to $\lambda=0$. The corresponding Choi matrices read as follows

\begin{equation}\label{}
\mu=0 \to \left( \begin{array}{cc|cc} a & . & . & \lambda \\  . & 1-b & . & .  \\ \hline   . & . & 1-a & . \\   \overline{\lambda} & . & . & b \end{array} \right)  \ , \ \ \
\lambda =0 \to \left( \begin{array}{cc|cc} a & . & . & . \\  . & 1-b & \overline{\mu} & .  \\ \hline   . & \mu & 1-a & . \\   . & . & . & b \end{array} \right) .
\end{equation}
Let us observe that it is sufficient to check operator Schwarz inequality (\ref{S!}) for traceless operators only. Indeed, simple algebra shows that

\begin{equation}\label{S0}
  \Phi(X^\dagger X) - \Phi(X^\dagger)\Phi(X)  =  {\Phi(X_0^\dagger X_0) } - \Phi(X_0^\dagger)\Phi(X_0) ,
\end{equation}
where $X=X_0 + z \oper$ and ${\rm Tr}X_0=0$.

\begin{proposition} \label{PRO-S-I} A unital phase covariant map (\ref{!}) satisfies Schwarz inequality if and only if

\begin{equation}\label{S-I}
  |\lambda| \leq \min\{\sqrt{a},\sqrt{b}\} .
\end{equation}
\end{proposition}

Proof: note that (\ref{S-I}) is necessary. Indeed, if $X= |1\>\<2|$, then $X^\dagger X = |2\>\<1|$ and hence one finds
\begin{equation}\label{}
  M = (b- |\lambda|^2) |2\>\<2|  .
\end{equation}
Positivity of $M$ requires $|\lambda| \leq \sqrt{b}$. Similarly, for $X= |2\>\<1|$ one obtains $|\lambda| \leq \sqrt{a}$. Surprisingly, these conditions are not only necessary but already sufficient. To show this let as assume that $a \leq b$ and take $|\lambda| = \sqrt{a}$. Consider a traceless matrix
\begin{equation}\label{X0}
	X = \left( \begin{array}{cc} z_0 & z_1 \\ z_2 & - z_0 \end{array} \right) ,
\end{equation}
with $z_0,z_1,z_2 \in \mathbb{C}$. It is clear that one can freely multiply $X$ by a complex number and hence one can always fix $f:=z_0$ to be real. If $X$ is diagonal, that is, $z_1=z_2=0$, one gets

\begin{equation}\label{X0}
	M = 4 f^2 \left( \begin{array}{cc} a(1-a) & 0 \\ 0 & b(1-b) \end{array} \right) ,
\end{equation}
which is obviously positive definite since $0\leq a,b\leq 1$. Now, if $z_1,z_2\neq 0$, there are two alternatives: $f =0$ or $f \neq 0$. If $f=0$ one obtains

\begin{equation}\label{}
  M = \left( \begin{array}{cc}
 (1-a) |z_1|^2 + (a-|\lambda|^2)|z_2|^2 & 0 \\
 0 & (1-b) |z_2|^2 + (b-|\lambda|^2)|z_1|^2
\end{array} \right) ,
\end{equation}
and hence $M \geq 0$ if and only if (\ref{S-I}) holds. The last scenario corresponding to $f \neq 0$ does not generate any new constraints. Indeed, assume that $a \leq b$ and check the worst case scenario corresponding to $|\lambda|=\sqrt{a} = \min\{\sqrt{a},\sqrt{b}\}$. One finds for the $M$ matrix

\begin{equation}\label{}
 M = \left(
\begin{array}{cc}
 (1-a) (4 f^2 a + |z_1|^2) &  {2f \sqrt{a} ( [1-a] z_1 - [1-b] \overline{z}_2 ) }\\
  {2f \sqrt{a} ( [1-a] \overline{z}_1 - [1-b] z_2)} & (1-b) (4 f^2 b + |z_2|^2) + (b-a)|z_1|^2 \end{array}
\right) .
\end{equation}
Note, that diagonal elements are non-negative and

\begin{equation}\label{}
  {\rm det}\, M = (1-a)(1-b) A + (b-a)\Big[ 4(1-a)(1-b)f^2|z_1|^2 + (1-a)|z_1|^4 + 4a (1-b)f^2 |z_2|^2\Big]  ,
\end{equation}
with

\begin{equation}\label{}
  A = 16 ab f^4 + |z_1|^2 |z_2|^2 + 4af^2(z_1 z_2 + \overline{z_1 z_2} ) .
\end{equation}
Now, since $ab \geq a^2$ one has

\begin{equation}\label{}
  A \geq 16 a^2 f^4 + |z_1|^2 |z_2|^2 + 4af^2(z_1 z_2 + \overline{z_1 z_2} ) = |4 a f^2 + z_1 z_2 |^2 ,
\end{equation}
and hence ${\rm det}\, M \geq 0$ which ends the proof of Proposition \ref{PRO-S-I}.   \hfill $\Box$

In the bi-stochastic case, i.e. $b=a$ one finds

\begin{equation}\label{Mb=a}
  {\rm det}\, M = (1-a)^2 |4 a f^2 + z_1 z_2 |^2  .
\end{equation}

\begin{remark} \label{R-f=0} Note, that the minimum of ${\rm det}\, M$ can be always realized with $f=0$. Indeed, if $f=0$ and $z_1 =0$, one finds ${\rm det}\, M=0$ for all $z_2$. It shows that the class with $f=0$ is critical to check Schwarz property.
\end{remark}

\begin{corollary} A unital conjugate phase covariant map (\ref{!}) satisfies Schwarz inequality if and only if

\begin{equation}\label{S-II}
  |\mu| \leq \min\{\sqrt{1-a},\sqrt{1-b}\} .
\end{equation}

\end{corollary}

\begin{Example} The Choi matrices corresponding to the transposition and reduction maps read

\begin{equation}\label{CHOI}
  C_T = \left( \begin{array}{cc|cc} 1 & . & . & . \\  . & . & 1 & .  \\ \hline   . & 1 & . & . \\   . & . & . & 1 \end{array} \right) , \ \ \
  C_R = \left( \begin{array}{cc|cc} . & . & . & -1 \\  . & 1 & . & .  \\ \hline   . & . & 1 & . \\   -1 & . & . & . \end{array} \right) .
\end{equation}
Hence, neither $T$ nor $R$ is Schwarz. The first example of a Schwarz map in $\mathcal{M}_2$ which is not 2-positive was provided by Choi \cite{Choi2}

\begin{equation}\label{}
  \Phi(X) = \frac 14 \oper {\rm Tr}X + \frac 12 X^{\rm T} .
\end{equation}
The corresponding Choi matrix has the following form

\begin{equation}\label{CHOI}
  C = \left( \begin{array}{cc|cc} 3/4 & . & . & . \\  . & 1/4 & 1/2 & .  \\ \hline   . & 1/2 & 1/4 & . \\   . & . & . & 3/4 \end{array} \right) ,
\end{equation}
and hence this map belongs to our class with $a=b=3/4$, $\lambda=0$, and $\mu = \sqrt{1-a} = 1/2$.
\end{Example}

\begin{remark}

In a recent paper \cite{kS} the authors analyzed two classes of maps $\Psi_+,\Phi_- : \mathcal{M}_n \to \mathcal{M}_n$ defined as follows

\begin{equation}\label{}
  \Phi_-(X) :=    \frac{1}{{\rm Tr}A -1}\,( \oper \, {\rm Tr}(AX)  - X)  , \ \ \
  \Psi_+(X) :=    \frac{1}{{\rm Tr}A +1}\,( \oper \, {\rm Tr}(AX)  + X^\rT) ,
\end{equation}
with $A \in \mathcal{M}_n$. Note, that for $n=2$ and $A = \lambda_1 P_1 + \lambda_2 P_2$, the map $\Phi_-$ is phase-covariant and $\Psi_+$ is conjugate phase-covariant. One proves \cite{kS}

\begin{theorem} \label{TH-1}  $\Phi_-$ is Schwarz  if and only if $A^{-1}$ exists and $\| A^{-1} \|_\infty \leq \frac{{\rm Tr}A-1}{{\rm Tr}A}$.
\end{theorem}

\begin{theorem}  \label{TH-2}

If ${\rm Tr}A > -1$, then $\Psi_+$ is Schwarz if and only if $A \geq \frac{1}{{\rm Tr} A+1} \oper$ .

\end{theorem}
In Appendix we show that for $A = \alpha_1 P_1 + \alpha_2 P_2$ Theorem \ref{TH-1} is equivalent to Proposition (\ref{S-I})  and Theorem \ref{TH-2} is equivalent to (\ref{S-II}). It should be stressed, however, that $\Phi_-$ and $\Psi_+$ are just special examples of maps from the family (\ref{!}). Interestingly, authors of \cite{kS} characterized the properties of a matrix $A$ which guarantee that maps $\Phi_-$ and $\Psi_+$ are $k$-Schwarz. Recall, that $\Phi : \mathcal{M}_n \to \mathcal{M}_n$ is $k$-Schwarz if ${\rm id}_k \otimes \Phi$ is Schwarz. For $n=2$ the only nontrivial case corresponds to $k=1$ since a map which is 2-Schwarz is already 2-positive.

\end{remark}

\section{Schwarz maps with diagonal orthogonal symmetry}   \label{S-IV}

In this section we show how to generalize conditions (\ref{S-I}) and (\ref{S-II}) for maps satisfying (\ref{III}), i.e. represented by (\ref{!}).
We already found that both positivity, complete positivity, and Schwarz property for maps with diagonal unitary symmetry is fully characterized in terms of $|\lambda|$ and $|\mu|$. The same property holds for diagonal orthogonal symmetry due to the following

\begin{proposition} A unital map $\Phi$ parameterized by $(a,b,\lambda,\mu)$ is Schwarz if and only if a map parameterized by $(a,b,|\lambda|,|\mu|)$ is Schwarz.
\end{proposition}
Proof: A unital map $\Phi$ is Schwarz if $\widetilde{\Phi}(X) := U\Phi(V X V^\dagger) U^\dagger$ is Schwarz for arbitrary unitaries $U$ and $V$. Now, for the class of maps (\ref{!}) one easily finds diagonal unitaries $U$ and $V$ such that if $\Phi$ parameterized by $(a,b,\lambda,\mu)$, then $\widetilde{\Phi}$ parameterized by $(a,b,|\lambda|,|\mu|)$. \hfill $\Box$.

Now, it is easy to find the following necessary conditions

\begin{proposition}\label{Sch}
    If $\Phi$ is Schwarz then $\vert \lambda \vert \leq {\rm min} \lbrace \sqrt{a}, \sqrt{b} \rbrace $ and $\vert \mu \vert \leq {\rm min} \lbrace \sqrt{1-a}, \sqrt{1-b} \rbrace$.
\end{proposition}

\begin{proof}
    Consider $X= |1\>\<2|$. One finds
    \begin{equation}\label{SN1}
        M = \Phi (X^\dagger X)-\Phi(X^\dagger)\Phi(X)= \left( \begin{array}{cc} (1-a)-\vert \mu \vert^2 & 0 \\ 0 & b-\vert \lambda \vert^2  \end{array} \right)
    \end{equation}
    Now if $\Phi$ is Schwarz then (\ref{SN1}) is positive semidefinite and hence $\vert \mu \vert^2 \leq (1-a)$ and $\vert \lambda \vert^2 \leq b $.     Similarly, considering $X=|2\>\<1|$ one obtains

    \begin{equation}\label{SN2}
       M= \Phi (X^\dagger X)-\Phi(X^\dagger)\Phi(X)= \left( \begin{array}{cc} a-\vert \lambda \vert^2 & 0 \\ 0 & (1-b)-\vert \mu \vert^2  \end{array} \right)
\end{equation}
If $\Phi$ is Schwarz, then (\ref{SN2}) is positive semidefinite which gives $\vert \lambda \vert^2 \leq a$ and $\vert \mu \vert^2 \leq (1-b) $. \end{proof}
The above condition are necessary but not sufficient. For example, assuming $a \leq b$,  the map corresponding to $|\lambda| = \sqrt{a}$ and $|\mu|=\sqrt{1-b}$ is not even positive since

$$  |\lambda| + |\mu| = \sqrt{a} + \sqrt{1-b} > \sqrt{a b} + \sqrt{(1-a)(1-b)} , $$
and hence the condition (\ref{P1}) is violated.

The main result of this paper consists in the following

\begin{theorem} \label{TH} A unital map (\ref{!}) satisfies Schwarz inequality if and only if

\begin{equation}\label{elipse}
  \frac{|\lambda|^2}{a} +  \frac{|\mu|^2}{1- a} \leq 1 , \ \ \ {\rm and} \ \ \
 \frac{|\lambda|^2}{b} +  \frac{|\mu|^2}{1- b} \leq 1 .
\end{equation}
\end{theorem}

\begin{remark}
Note, that for a given pair $(a,b)$ a set of maps satisfying (\ref{elipse}) is strictly larger than a set of completely positive maps represented by a green rectangle on Figure \ref{FIG}. The above elliptic conditions (\ref{elipse}) show how much one can go beyond this rectangle:

\begin{eqnarray}\label{elipse-12}
  && \frac{|\lambda|^2}{a} +  \frac{|\mu|^2}{1- a} \leq 1 , \ \ \ {\rm for} \ \ |\mu| \leq \sqrt{(1-a)(1-b)}\ ,  \\
&& \frac{|\lambda|^2}{b} +  \frac{|\mu|^2}{1- b} \leq 1 , \ \ \ {\rm for} \ \ |\lambda| \leq \sqrt{ab} \ .
\end{eqnarray}

\end{remark}
Proof of the theorem: for diagonal $X$ the proof is the same as of Proposition \ref{PRO-S-I} since in this case $\lambda$ and $\mu$ are irrelevant. For $X$ defined by (\ref{X0}) with $z_0=0$ one obtains

\begin{equation}\label{}
       M = \left( \begin{array}{cc} (1-a)|z_1|^2 + a|z_2|^2 - |\mu z_1 + \lambda z_2|^2  & 0 \\ 0 & (1-b)|z_2|^2 + b |z_1|^2 - |\lambda z_1 + \mu z_2|^2   \end{array} \right) .
\end{equation}
It is clear that $M \geq 0$ if and only if both diagonal elements are non-negative. Consider $M_{11} = (1-a)|z_1|^2 + a|z_2|^2 - |\mu z_1 + \lambda z_2|^2$. The worst case scenario corresponds to $z_1$ and $z_2$ such that

\begin{equation}\label{}
  |\mu z_1 + \lambda z_2| = |\mu| |z_1| + |\lambda| |z_2| .
\end{equation}
Now, let $|z_2| = x |z_1|$ with $x \geq 0$. One finds

\begin{equation}\label{}
  M_{11} = |z_1|^2 \Big( 1-a + a x^2 - (|\mu| + x |\lambda|)^2 \Big) ,
\end{equation}
and hence positivity of $M_{11}$ gives rise to the following quadratic inequality for the real parameter `$x$'

\begin{equation}\label{}
  (a-|\lambda|^2) x^2 - 2 |\lambda||\mu|x + (1-a-|\mu|^2) \geq 0 .
\end{equation}
The above inequality holds if and only if the corresponding discriminant

\begin{equation}\label{}
  \Delta = 4  |\lambda|^2|\mu|^2 - 4 (a-|\lambda|^2)(1-a-|\mu|^2) = 4a(1-a)\Big(\frac{|\lambda|^2}{a} +  \frac{|\mu|^2}{1- a} -1 \Big)  ,
\end{equation}
satisfies $\Delta \leq 0$, which is equivalent to

\begin{equation}\label{}
  \frac{|\lambda|^2}{a} +  \frac{|\mu|^2}{1- a} \leq 1 .
\end{equation}
Similarly, $M_{22}\geq 0$ if and only if

\begin{equation}\label{}
  \frac{|\lambda|^2}{b} +  \frac{|\mu|^2}{1- b} \leq 1 .
\end{equation}
Hence, the class of $X$ with $f=0$ already give rise to (\ref{elipse}). Clearly, if $\mu=0$ they reduce to $|\lambda|\leq \min\{\sqrt{a},\sqrt{b}\}$, and if $\lambda=0$ they reduce to $|\mu|\leq \min\{\sqrt{1-a},\sqrt{1-b}\}$.

Now, adding nontrivial diagonal elements (i.e. $f\neq 0$) one does not produce any new constraints.  One finds for the $M$ matrix

\begin{eqnarray*}
  M_{11} &=& 4a(1-a)f^2 + (1-a)|z_1|^2 + a|z_2|^2 -|\lambda \overline{z}_2 + \mu \overline{z}_1|^2 , \\
   M_{22} &=& 4b(1-b)f^2 + (1-b)|z_2|^2 + b|z_1|^2 -|\lambda {z}_1 + \mu {z}_2|^2 , \\
    M_{12} &=&{ -2f \left\{ \lambda\Big(  [1-b] \overline{z}_2  -[1-a] z_1\Big) + \overline{\mu} \Big( a z_2 -b \overline{z}_1\Big)  \right\} } = \overline{M_{21}} .
\end{eqnarray*}
Suppose that one of the ellipse conditions (\ref{elipse}) is saturated, e.g.

\begin{equation}\label{elipse-1}
  \frac{|\lambda|^2}{a} +  \frac{|\mu|^2}{1- a} = 1 ,
\end{equation}
that is,

\begin{equation}\label{elipse-1}
  |\mu| = \sqrt{ (1-a)\left( 1 - \frac{|\lambda|^2}{a}\right) }   .
\end{equation}
Let us observe that the diagonal elements $M_{11}$ and $M_{22}$ are non-negative. Indeed, the diagonal elements are the same as in (\ref{SN2}) but shifted by $4a(1-a)f^2$ for $M_{11}$ and $4b(1-b)f^2$ for $M_{22}$, respectively. Hence, to prove positivity of $M$ it is sufficient to show that ${\rm det}M \geq 0$. If $b=a$ (bistochastic case) after a  tedious algebra one finds

\begin{equation}\label{}
  {\rm det} M = \Big| 4a(1-a)f^2 + \Big( \frac{a^2-2a|\lambda|^2 + |\lambda|^2}{a} z_1 z_2 - |\lambda| \sqrt{ \frac{1-a}{a}} \sqrt{a - |\lambda|^2} (z_1^2 + z_2^2) \Big) \Big|^2 .
\end{equation}
Note, that if $\mu=0$, that is, $|\lambda| = \sqrt{a}$, the above formula reduces to (\ref{Mb=a}). For $b\neq a$ the analysis of ${\rm det} M$ becomes quite complex but numerical analysis perfectly confirms that if the ellipse condition (\ref{elipse}) is saturated then the minimal value of ${\rm det} M$ equals to zero and is realized by $X$ with $f=0$.

For any fixed $a,b \in [0,1]$ the corresponding regions of positivity, Schwarz property and complete positive can be nicely represented on the $(|\lambda|,|\mu|)$ plane,  cf. Figure \ref{FIG}.

\begin{figure}[!htb]
\minipage{0.32\textwidth}
  \includegraphics[width=5.1cm]{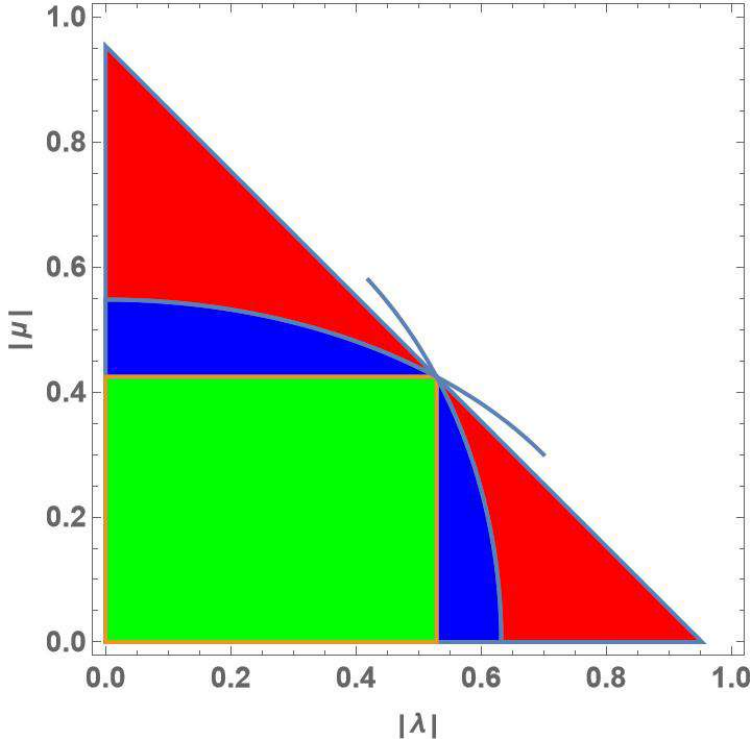}
\endminipage\hfill
\minipage{0.32\textwidth}
  \includegraphics[width=5.1cm]{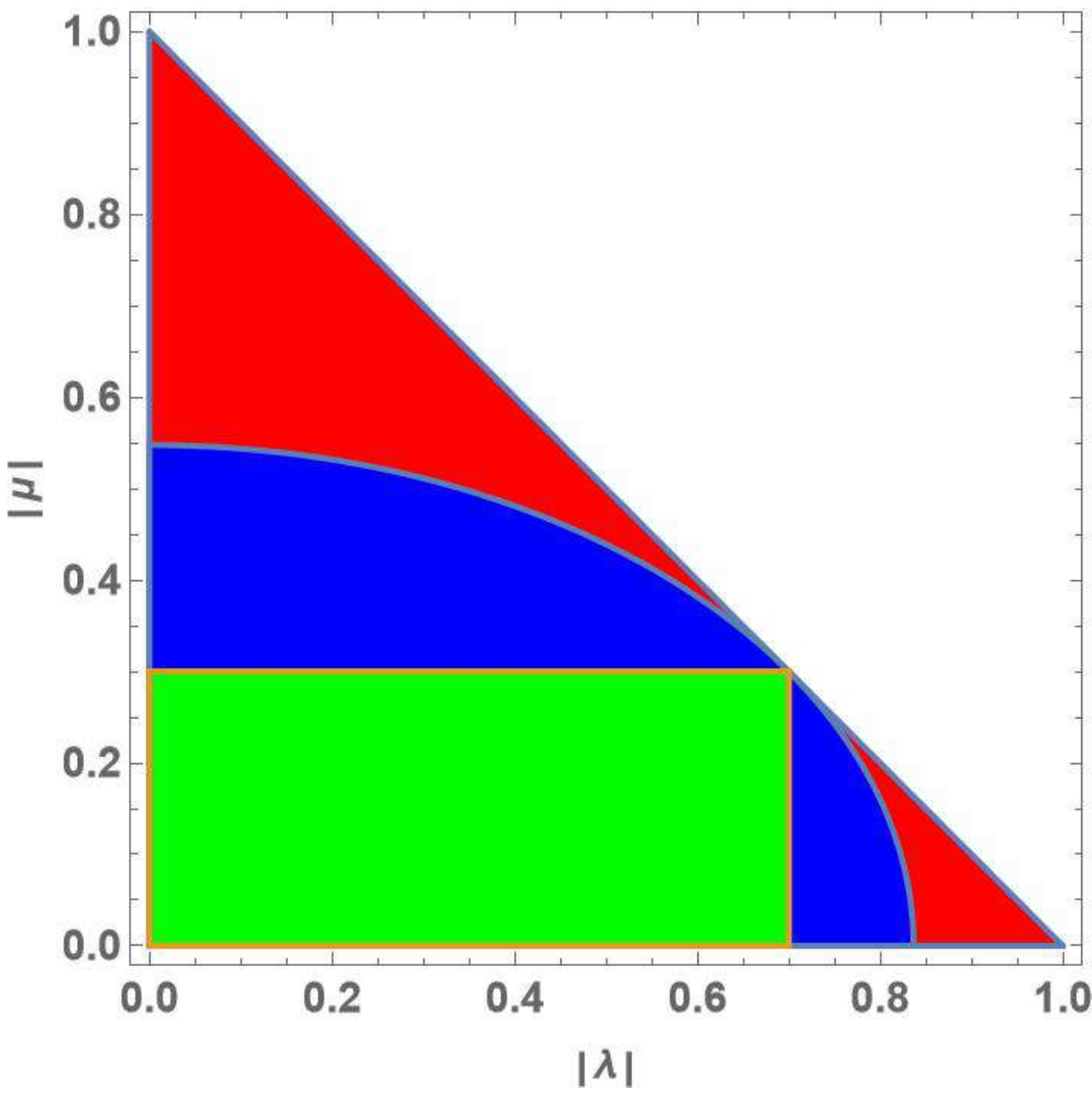}
\endminipage\hfill
\minipage{0.32\textwidth}%
  \includegraphics[width=5.1cm]{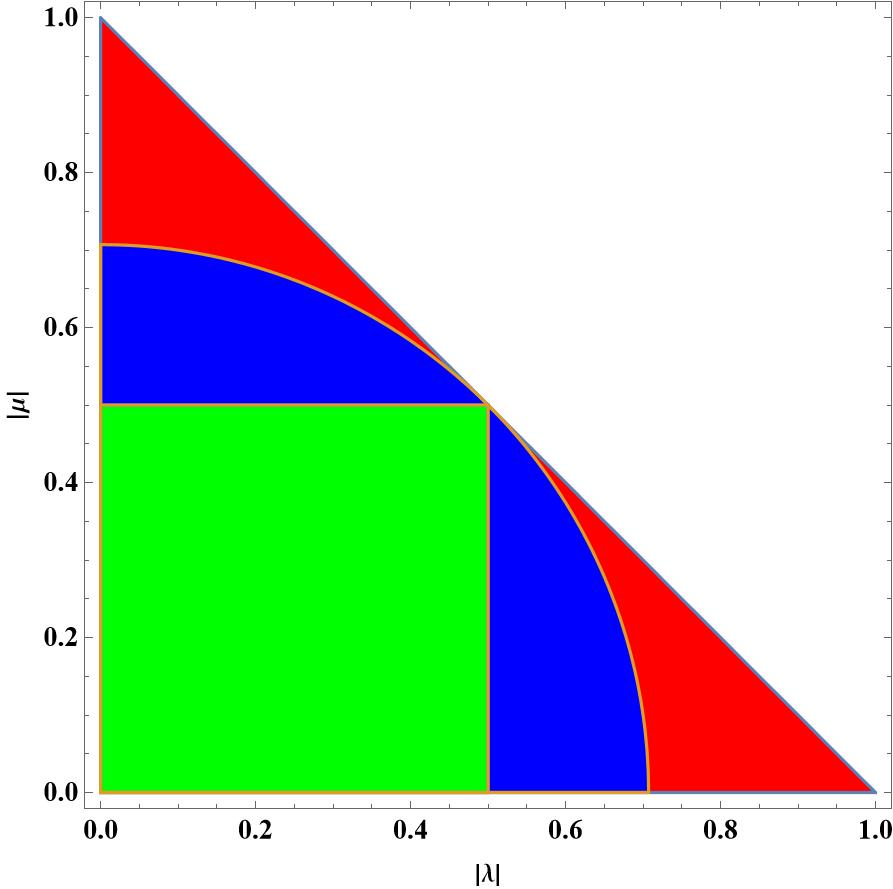}
\endminipage
\caption{(Color online) Green regions represent completely positive maps; blue regions represent Schwarz maps which are not completely positive; red regions represent positive maps which are not Schwarz. {Left plot: $a=0.7 $ $b = 0.4$} --- two ellipses intersect at $(\sqrt{ab},\sqrt{(1-a)(1-b)})$. {Middle plot: $a=b=0.7 $} (bistochastic) two ellipses reduce to a single ellipse which touches the border of positive maps at $(a,1-a)$. Right plot: $a=b=0.5$ --- a single ellipse becomes a circle and a green rectangle of completely positive maps becomes a square.   \label{FIG}}
\end{figure}

\section{Pauli maps: a case study}   \label{S-Pauli}


Pauli maps provide a prominent example of maps from the family (\ref{!}). Any Pauli map can be represented via

\begin{equation}\label{Pauli}
  \Phi(X) = \sum_{\alpha=0}^3 p_\alpha \sigma_\alpha X \sigma_\alpha ,
\end{equation}
with $p_\alpha \in \mathbb{R}$. Note that a Pauli map is trace-preserving (and at the same time  unital) if $\sum_\alpha p_\alpha =1$. The corresponding Choi matrix reads

\begin{equation}\label{}
  C = \left( \begin{array}{cc|cc} p_0+p_3 & . & . & p_0-p_3 \\  . & p_1+p_2 & p_1-p_2 & .  \\ \hline   . & p_1-p_2 & p_1+p_2 & . \\  p_0-p_3 & . & . & p_0+p_3 \end{array} \right) ,
\end{equation}
and hence the map corresponds to

$$   a=b = p_0+p_3\ , \ \ \lambda = p_0-p_3 \ , \ \ \mu = p_1 - p_2 .   $$
The map is completely positive iff $p_\alpha \geq 0$. It is positive iff

\begin{equation}\label{}
  |p_0-p_3| + |p_1-p_2| \leq 1 ,
\end{equation}
and Schwarz iff

\begin{equation}\label{pppp}
  \frac{(p_0-p_3)^2}{p_0+p_3} + \frac{(p_1-p_2)^2}{p_1+p_2} \leq 1 .
\end{equation}
Note, that for a Pauli map if we fix $a\in [0,1]$, then the map is fully characterized by two parameters, e.g. $(p_0,p_1)$:

$$   p_3 = a - p_0 \ ,  \ \ \ p_2 = 1-a - p_1 . $$
Then the ellipse condition (\ref{pppp}) can be rewritten as follows

\begin{equation}\label{p2p2}
  \frac{\left(p_0 - \frac a2 \right)^2}{\frac{a}{4}} +   \frac{\left(p_1 - \frac{1-a}{2}\right)^2}{\frac{1-a}{4}} \leq 1 ,
\end{equation}
and can be nicely represented on the  $(p_0,p_1)$-plane, cf. Fig. \ref{FIG2}.

\begin{center}
\begin{figure}
\includegraphics[width=6cm]{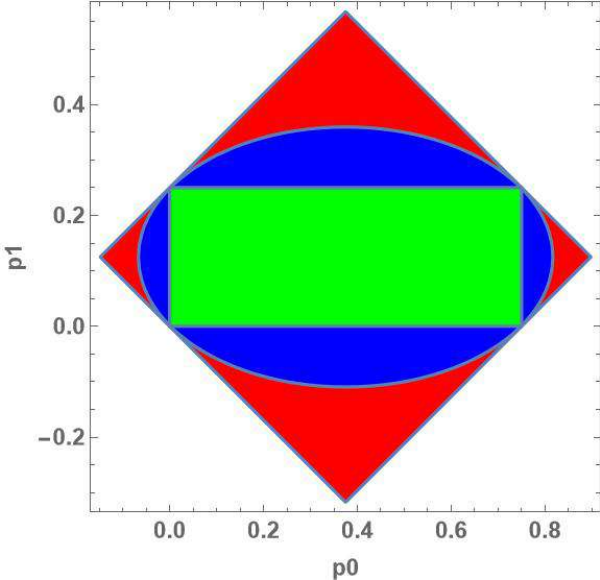}
\hspace*{.4cm} \includegraphics[width=6cm]{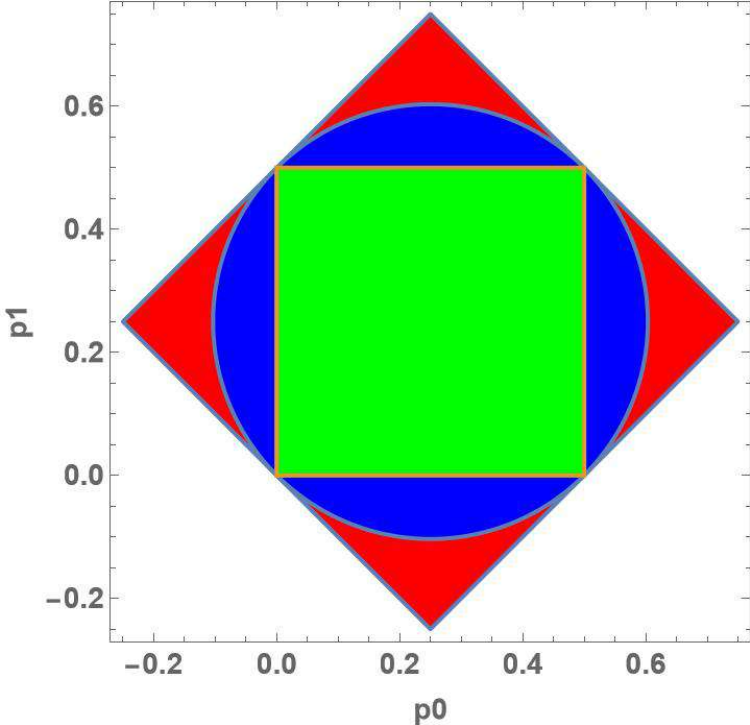} 
\caption{(Color online) Pauli maps: green rectangle represent completely positive maps; blue regions represent Schwarz maps which are not completely positive; red regions represent positive maps which are not Schwarz. {Left plot: $a=0.75 $. }Right plot: $a=0.5$ -- ellipse becomes a circle center at $(1/4,1/4)$ with radius $1/\sqrt{8}$,  and a green rectangle becomes square.    \label{FIG2} }
\end{figure}
\end{center}

Let us recall that positivity and complete positivity of a Pauli map can be fully characterized in terms of its spectrum: one finds

\begin{equation}\label{}
  \Phi(\sigma_\alpha) = \lambda_\alpha \sigma_\alpha , \ \ \ \ \alpha = 0,1,2,3 ,
\end{equation}
where the corresponding eigenvalues read

\begin{equation}\label{}
  \lambda_\alpha = \sum_{\beta=0}^3 H_{\alpha\beta} p_\beta ,
\end{equation}
and $H_{\alpha\beta}$ is a Hadamard matrix

\begin{equation}\label{H}
  H = \left( \begin{array}{cccc} 1 & 1 & 1 & 1 \\  1 & 1 & -1 & -1  \\    1 & -1 & 1& -1 \\  1 & -1 & -1 & 1 \end{array} \right) .
\end{equation}
Note that $\lambda_0=1$ corresponds to the fact that $\Phi$ is trace-preserving. One has

\begin{itemize}
  \item $\Phi$ is positive iff

\begin{equation}\label{Pos}
  |\lambda_k|\leq 1 \ , \ \ (k=1,2,3) ,
\end{equation}

  \item $\Phi$ is completely positive iff

  \begin{equation}\label{FA}
    |\lambda_1 \pm \lambda_2| \leq |1\pm \lambda_3| .
  \end{equation}
\end{itemize}
Condition (\ref{FA}) was derived by  Algoet and Fujiwara \cite{FA}. Again, both conditions for positivity (\ref{Pos}) and Fujiwara-Algoet conditions for complete positivity have simple geometric interpretation \cite{KAROL,Mario,Karol}. They define a tetrahedron in the cube with vertices $\lambda_k=\pm 1$. The four vertices of the tetrahedron: $(1,-,1-1), (-1,1,-1), (1,-1,-1)$ and $(1,1,1)$, cf. Fig. \ref{FIG3}.

\begin{center}
\begin{figure}
 \centering
\includegraphics[width=6cm]{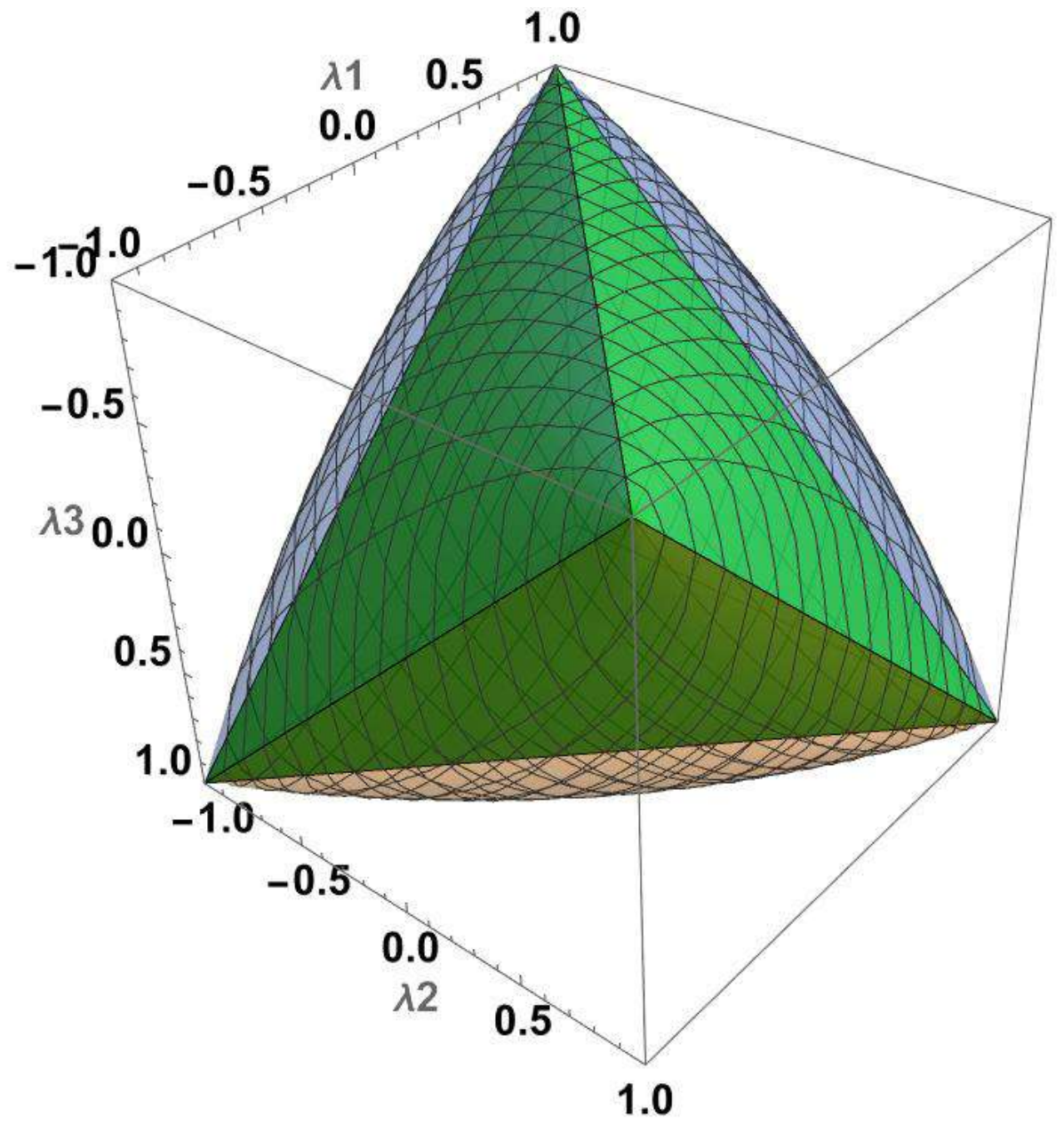} 
\caption{(Color online)
Fujiwara-Algoet tetrahedron (\ref{FA}) within a Fujiwara-Algoet-Schwarz surface (\ref{lll}).
  \label{FIG3} }
\end{figure}
\end{center}

Now, condition (\ref{p2p2}) can be rewritten in terms of $\{\lambda_k\}_{k=1}^3$ as follows

\begin{equation}\label{FAS}
  (1-\lambda_3)(\lambda_1 + \lambda_2)^2 + (1+\lambda_3)(\lambda_1 - \lambda_2)^2 \leq 2 .
\end{equation}
Note, that (\ref{FA}) immediately implies (\ref{FAS}). The boundary of a convex set of Schwarz Pauli maps is represented by

\begin{equation}\label{lll}
  \lambda_3 = \lambda_1 \lambda_2 \pm \sqrt{(1-\lambda_1^2)(1-\lambda_2^2)} ,
\end{equation}
and it is displayed on Fig. \ref{FIG3}. Note, that four vertices of the tetrahedron satisfy (\ref{lll}) and hence they belong to the boundary of the set of Schwarz maps.
The maximal violation of Fijiwara-Algoet conditions (\ref{FA}) corresponds to $\lambda_1=\lambda_2=\lambda_3 = - \frac 12$, or, equivalently, $p_0=-\frac 18$ and $p_1=p_2=p_3 = \frac 38$. For positive Pauli maps the maximal departure for complete positivity corresponds to vertex $(-1,-1,-1)$, or, equivalently, $p_0=-\frac 12$ and $p_1=p_2=p_3 = \frac 12$, that is, the Pauli map

\begin{equation}\label{}
  \Phi(X) = \frac 12 (\sigma_1 X \sigma_1 + \sigma_2 X \sigma_2 + \sigma_3 X \sigma_3 - X) ,
\end{equation}
is positive but not Schwarz, and

\begin{equation}\label{}
  \Phi(X) = \frac 38 \left(\sigma_1 X \sigma_1 + \sigma_2 X \sigma_2 + \sigma_3 X \sigma_3 - \frac 13 X\right) ,
\end{equation}
is Schwarz but evidently not completely positive.

\begin{corollary} A Pauli map

\begin{equation}\label{}
  \Phi_\alpha(X) := \frac{1}{3-\alpha} (\sigma_1 X \sigma_1 + \sigma_2 X \sigma_2 + \sigma_3 X \sigma_3 - \alpha X) ,
\end{equation}
is positive iff $\alpha \leq 1$, Schwarz iff $\alpha \leq \frac 13$, and completely positive iff $\alpha \leq 0$.
\end{corollary}
One may call (\ref{FAS}) the Fujiwara-Algoet-Schwarz condition. Note that (\ref{FAS}) defines a 3-dimensional region containing the tetrahedron of Pauli channels, cf. Figure \ref{FIG3}. One easily computes the corresponding volumes

\begin{equation}\label{}
V_{\rm Positive} = 8 \ , \ \ \   V_{\rm Schwarz} =  \frac{\pi^2}{2} \approx 4.93  \ , \ \ \  V_{\rm CP} = 4 \ ,
\end{equation}
which shows that the volume of a set of Pauli Schwarz maps is $23\%$ larger than  of Pauli channels, and a set of positive Pauli maps is $62\%$ larger than a set of Schwarz maps.

\section{Generalized Schwarz covariant maps}   \label{S-IVa}

In a recent paper \cite{Alex} authors proposed the following

\begin{definition} A linear map $\Phi : \mathcal{M}_n \to \mathcal{M}_m$ is called a generalized Schwarz map if

\begin{equation}\label{}
  \left( \begin{array}{cc} \Phi(\oper_n) & \Phi(K) \\ \Phi(K)^\dagger & \Phi(K^\dagger K) \end{array} \right) \geq 0 ,
\end{equation}
for all $K \in \mathcal{M}_n$ .
\end{definition}
One immediately shows that $\Phi$ is generalized Schwarz if and only if

\begin{equation}\label{GS}
  \Phi(X^\dagger X) \geq \Phi(X)^\dagger \Phi(\oper_n)^+ \Phi(X) ,
\end{equation}
where $A^+$ denotes a generalized Moore-Penrose inverse of $A$. Actually, such class of maps was already analyzed in \cite{KS} (in \cite{KS} it was assumed that $\Phi(\oper_n)$ is strictly positive). Note, that if the map is unital (\ref{GS}) reduces to the original Schwarz inequality. Interestingly $\Phi$ is generalized Schwarz if and only if it satisfies the following tracial inequality \cite{Alex}

\begin{equation}\label{}
{\rm Tr} [ \Phi^\ddag(K^\dagger X^+ K)] \geq {\rm Tr}[\Phi^\ddag(K^\dagger) \Phi^\ddag(X)^+ \Phi^\ddag(K) ] ,
\end{equation}
for all $K \in \mathcal{M}_n$, $X \geq 0$ such that ${\rm Ker}\, X \subseteq {\rm Ker}\, K^\dagger$.

Consider now a map $\Phi$ defined by (\ref{!}) which is in general not unital. Let us use the following notation

\begin{equation}\label{A'}
   \left( \begin{array}{cc} a_{11} & a_{12} \\ a_{21} & a_{22}  \end{array} \right) =
    \left( \begin{array}{cc} a & a' \\ b' & b  \end{array} \right) \ , \ \ \ \ a,b,a',b'\geq 0 .
\end{equation}

\begin{proposition} $\Phi$ is generalized Schwarz if and only if

\begin{equation}\label{elipse-G}
  \frac{|\lambda|^2}{a} +  \frac{|\mu|^2}{a'} \leq B , \ \ \  \frac{|\lambda|^2}{b} +  \frac{|\mu|^2}{b'} \leq A .
\end{equation}
where $A:= a+a'$ and $B:=b+b'$.
\end{proposition}
Proof: note that $\Phi$ is generalized Schwarz if and only if the following map

\begin{equation}\label{GenS}
  \widetilde{\Phi}(X) := \sqrt{\Phi(\oper)^+}\, \Phi(X)\, \sqrt{\Phi(\oper)^+} ,
\end{equation}
is Schwarz. One finds

\begin{equation}\label{}
  \Phi(\oper) =  \left( \begin{array}{cc} A& 0 \\ 0 & B  \end{array} \right) ,
\end{equation}
and hence $ \widetilde{\Phi}$ belongs to the class (\ref{!}) with the corresponding parameters:

\begin{equation}\label{}
  \widetilde{a} := \frac{a}{A} \ , \ \ \widetilde{a'} := \frac{a'}{A} \ , \ \  \widetilde{b} := \frac{b}{B} \ , \  \ \widetilde{b'} := \frac{b'}{B} ,
\end{equation}
and

\begin{equation}\label{}
  \widetilde{\lambda} := \frac{\lambda}{\sqrt{AB}} \ , \  \ \ \ \widetilde{\mu} := \frac{\mu}{\sqrt{AB}} .
\end{equation}
Since $\widetilde{\Phi}$ is unital it is Schwarz if and only if the corresponding parameters satisfy (\ref{elipse}). The ellipse conditions (\ref{elipse}) are equivalent to (\ref{elipse-G}) in terms of the original parameters of $\Phi$. \hfill $\Box$

Concerning Schwarz inequality there is an interesting relation between $\Phi$ and its dual $\Phi^\ddag$ defined via

\begin{equation}\label{}
  {\rm Tr} (\Phi^\ddag(X^\dagger) Y) := {\rm Tr}(X^\dagger \Phi(Y)) .
\end{equation}
If $\Phi$ belongs to the family (\ref{!}) its dual also belongs to (\ref{!}), namely

\begin{equation}\label{!!}
  \Phi^\ddag(X) = \sum_{i,j=1}^2 a^{\rm D}_{ij} |i\>\<j| X |j\>\<i| + \Big( \overline{\lambda} P_1 X P_2 + {\lambda} P_2 X P_1 \Big) + \Big( \mu P_1 X^{\rm T}  P_2 + \overline{\mu} P_2 X^{\rm T}  P_1 \Big) ,
\end{equation}
with $a^{\rm D}_{ij} := a_{ji}$.

\begin{corollary} $\Phi^\ddag$ is generalized Schwarz if and only if

\begin{equation}\label{elipse-GD}
  \frac{|\lambda|^2}{a} +  \frac{|\mu|^2}{b'} \leq b+a' , \ \ \  \frac{|\lambda|^2}{b} +  \frac{|\mu|^2}{a'} \leq a+b' .
\end{equation}
\end{corollary}

\begin{proposition} \label{Pro-7} 
The following relations hold:

\begin{enumerate}
  \item for $a'=b'$ the map  $\Phi$ is generalized Schwarz if and only if $\Phi^{\ddag}$ is generalized Schwarz,

  \item for $a'>b'$ if $\Phi$ is generalized Schwarz then $\Phi^{\ddag}$ is also generalized Schwarz,

  \item for $a'<b'$ if $\Phi^\ddag$ is generalized Schwarz then $\Phi$ is also generalized Schwarz.

\end{enumerate}

\end{proposition}
For the proof cf. Appendix.

\begin{corollary} In the unital case, i.e. $a'=1-a$ and $b'=1-b$,   the following relations hold:

\begin{enumerate}
  \item for $a=b$ the map  $\Phi$ is Schwarz if and only if $\Phi^{\ddag}$ is  Schwarz,

  \item for $a<b$ if $\Phi$ is Schwarz then $\Phi^{\ddag}$ is  generalized Schwarz,

  \item for $a>b$ if $\Phi^\ddag$ is generalized Schwarz then $\Phi$ is Schwarz.

\end{enumerate}

\end{corollary}
This clearly shows that contrary to positivity and complete positivity the generalized Schwarz property is not preserved by passing to the dual map.


\section{Conclusions}  \label{S-C}

We provided a detailed analysis of positivity properties of a large class of qubit maps represented by (\ref{!}). This class of maps satisfies diagonal unitary or orthogonal symmetry and already found a lot of applications in the study of completely positive quantum dynamics and in detection of two-qubit quantum entanglement. Our analysis clearly shows that for a subclass of unital maps

$$   {\rm completely\ positive\ maps} \ \subsetneq\ {\rm Schwarz\ maps} \  \subsetneq \ {\rm positive\ maps} . $$
Our analysis can be  summarized as follows:

\begin{itemize}
  \item $\Phi$ is completely positive if and only if

  $$ |\lambda| \leq \sqrt{ab} \ ,\ \ {\rm and} \ \ |\mu|\leq  \sqrt{(1-a)(1-b)} . $$

  \item $\Phi$ is a Schwarz map   if and only if

  $$ \frac{|\lambda|^2}{a} +  \frac{|\mu|^2}{1- a} \leq 1 , \ \ {\rm and} \ \  \frac{|\lambda|^2}{b} +  \frac{|\mu|^2}{1- b} \leq 1 . $$

  \item $\Phi$ is  positive if and only if

  $$ |\lambda| + |\mu| \leq \sqrt{ab}+  \sqrt{(1-a)(1-b)} . $$

\end{itemize}
The corresponding regions were illustrated in Figure \ref{FIG}. The whole analysis is illustrated by a seminal class of qubit Pauli maps. In this case our conditions give rise to Fujiwara-Algoet-Schwarz condition (\ref{FAS}) which defines a class of Schwarz Pauli maps.

Interestingly, the map originally found by Choi \cite{Choi2} which is Schwarz but not completely positive corresponds to $a=b=3/4$, $\lambda=0$, and $\mu = 1/2$. We also generalize the analysis to include so called generalized Schwarz maps \cite{Alex} and find an interesting relations between $\Phi$ and its dual $\Phi^\ddag$ (Proposition \ref{Pro-7}).  Following \cite{Ion1,Ion2,Ion3} it would be interesting to extend the analysis of maps displaying unitary and orthogonal diagonal symmetry beyond qubit scenario.

\appendix

\section{Proof of Proposition \ref{Pro-P}}

Proof: $\Phi$ is positive if and only if the Choi matrix (\ref{CHOI}) is block positive. Consider a general case, i.e. not necessary unital. Following notation (\ref{A'}) one has diagonal elements $a,a',b,b'\geq 0$. Now, the map is completely positive iff

$$   |\lambda| \leq \sqrt{ab}\ ,\ \ \ |\mu|\leq \sqrt{a'b'} . $$
It is well known \cite{Topical,Guhne} that in the qubit case the Choi matrix of a positive map may have at most one negative eigenvalue. Let us assume that $ab < |\lambda|^2$ and $a'b' > |\mu|^2$.  Now, due the seminal Woronowicz result \cite{Wor} any block-positive matrix in $\mathcal{M}_2 \otimes \mathcal{M}_2$ is decomposable, that is,

\begin{equation}\label{CAB}
  C = A + B^\Gamma ,
\end{equation}
with $A,B \geq 0$, and $B^\Gamma$ denotes partial transposition. Again, it is sufficient to prove the result for $\lambda = |\lambda|$ and $\mu = |\mu|$. Let

$$ |\lambda| = \sqrt{ab} + \kappa \ , \ \ \ \kappa > 0 . $$
One finds that $C$ is block-positive if and only if there exists $x,y \geq 0$ such that

\begin{equation}\label{???}
 xy = |\mu|^2 \ , \ \ \  (a' -x)(b'-y) \geq \kappa^2 .
\end{equation}
Indeed, in this case one finds (\ref{CAB}) with

\begin{equation}\label{}
  A= \left( \begin{array}{cc|cc} a & . & . & \sqrt{ab} \\  . & y & |\mu| & .  \\ \hline   . & |\mu| & x & . \\  \sqrt{ab} & . & . & b \end{array} \right) ,
%
 B =  \left( \begin{array}{cc|cc} . & . & . & . \\  . & b'-y & \kappa & .  \\ \hline   . &  \kappa & a'-x & . \\  . & . & . & . \end{array} \right) .
\end{equation}
Now, since $y = |\mu|^2/x$, one has to check whether  there exist $x > 0$ such that $(a' -x)(b'-|\mu|^2/x) \geq \kappa^2$. Such $x$ exists if and only if the minimum $F_{\rm min}$ of $F(x) := (a' -x)(b'-|\mu|^2/x)$ satisfies

$$ F_{\rm min} \geq \kappa^2 = (|\lambda| - \sqrt{ab})^2 . $$
One finds $F'(x_0)=0$ for $x_0 = |\mu| \sqrt{a'/b'}$ and hence

$$ F_{\rm min} =   (a' -x_0)(b'-|\mu|^2/x_0) = (\sqrt{a'b'} - |\mu|)^2 . $$
Finally, the condition $F_{\rm min} \geq  (|\lambda| - \sqrt{ab})^2$ is equivalent to $|\lambda| + |\mu| \leq \sqrt{ab} + \sqrt{a'b'}$. \hfill $\Box$

\section{Analysis of Theorems \ref{TH-1} and \ref{TH-2}}

Theorems \ref{TH-1}: note, that for $A = \alpha_1 P_1 + \alpha_2 P_2$ the map $\Phi_-$ belongs to (\ref{!}) with

\begin{equation}\label{}
   a_{ij} = \frac{1}{\alpha_1+\alpha_2-1}  \left( \begin{array}{cc} \alpha_1-1 & \alpha_2 \\ \alpha_1 & \alpha_2-1  \end{array} \right) \ , \ \ \ \lambda= - \frac{1}{\alpha_1+\alpha_2-1} \ , \ \ \ \mu=0 .
\end{equation}
Hence, according to (\ref{S-I}) the map $\Phi_-$ is Schwarz if and only if
\begin{equation}\label{C1}
  |\lambda| =  \frac{1}{\alpha_1+\alpha_2 -1} \leq \min\left\{ \sqrt{ \frac{\alpha_1-1}{\alpha_1+\alpha_2 -1} } , \sqrt{ \frac{\alpha_2-1}{\alpha_1+\alpha_2 -1} } \right\} .
\end{equation}
This condition is  equivalent to $\| A^{-1} \|_\infty \leq \frac{{\rm Tr}A-1}{{\rm Tr}A}$, i.e.

\begin{equation}\label{C2}
 \max\left\{ 1/\alpha_1,1/\alpha_2\right\} \leq  \frac{\alpha_1+\alpha_2-1}{\alpha_1+\alpha_2} .
\end{equation}
Indeed, let $\alpha_1 \geq \alpha_2$. Then (\ref{C1}) gives

\begin{equation}\label{C1a}
   \frac{1}{\alpha_1+\alpha_2 -1} \leq  \sqrt{ \frac{\alpha_1-1}{\alpha_1+\alpha_2 -1} } ,
\end{equation}
and hence $\alpha_1\alpha_2 + \alpha_2^2 - 2 \alpha_2 - \alpha_1 \geq 0$. Similarly (\ref{C2}) gives

\begin{equation}\label{C2a}
 \frac{1}{\alpha_2} \leq  \frac{\alpha_1+\alpha_2-1}{\alpha_1+\alpha_2} ,
\end{equation}
and hence again $\alpha_1\alpha_2 + \alpha_2^2 - 2 \alpha_2 - \alpha_1 \geq 0$. Hence, Theorem \ref{TH-1} provides a special case of Proposition \ref{PRO-S-I}.

Theorems \ref{TH-2}: note, that for $A = \alpha_1 P_1 + \alpha_2 P_2$ the map $\Psi_+$ belongs to (\ref{!}) with

\begin{equation}\label{}
   a_{ij} = \frac{1}{\alpha_1+\alpha_2+1}  \left( \begin{array}{cc} \alpha_1+1 & \alpha_2 \\ \alpha_1 & \alpha_2+1  \end{array} \right) \ , \ \ \ \lambda= 0 \ , \ \ \ \mu= \frac{1}{\alpha_1+\alpha_2+1} \ .
\end{equation}
Hence, according to (\ref{S-II}) the map $\Psi_+$ is Schwarz if and only if
\begin{equation}\label{C1m}
  |\mu| =  \frac{1}{\alpha_1+\alpha_2 +1} \leq \min\left\{ \sqrt{ \frac{\alpha_1}{\alpha_1+\alpha_2 +1} } , \sqrt{ \frac{\alpha_2}{\alpha_1+\alpha_2 +1} } \right\} .
\end{equation}
This condition is  equivalent to $A \geq \frac{1}{{\rm Tr}A+1} \oper$, i.e.

\begin{equation}\label{C2m}
 \min\left\{ \alpha_1,\alpha_2\right\} \geq  \frac{1}{\alpha_1+\alpha_2+1} .
\end{equation}
Indeed, let $\alpha_1 \geq \alpha_2$. Then (\ref{C1m}) gives $   \alpha_2(\alpha_1+\alpha_2 +1) \geq 1$.  Similarly (\ref{C2m}) gives $   \alpha_2(\alpha_1+\alpha_2 +1) \geq 1$. Hence Theorem \ref{TH-2} is equivalent to (\ref{S-II}).

\section{Proof of Proposition \ref{Pro-7}}

If $a'=b'$, then conditions (\ref{elipse-G}) and (\ref{elipse-GD}) coincide and hence the result follows.  Now let us assume that  $a'>b'$ and $\Phi$ is generalized Schwarz. Using condition (\ref{elipse-G})

\begin{equation*}\label{}
\frac{\vert \lambda\vert^2}{a}+\frac{\vert \mu\vert^2}{a'} \leq b+b'\ ,
\end{equation*}
for $|\mu| \leq \sqrt{a' b'}$, and

\begin{equation*}\label{}
  \frac{\vert \lambda\vert^2}{b}+\frac{\vert \mu\vert^2}{b'} \leq a+a' ,
\end{equation*}
for $|\lambda| \leq \sqrt{a b} $, one finds

$$  \frac{\vert \lambda\vert^2}{a} \leq b+b' - \frac{\vert \mu\vert^2}{a'} , $$
and hence

$$  \frac{\vert \lambda\vert^2}{a}+\frac{\vert \mu\vert^2}{b'} \leq b+b' + {\vert \mu\vert^2}  \frac{a'-b'}{a'b'} \leq b+a' ,$$
due to $|\mu|^2 \leq a' b'$, and $a'-b'\geq 0$. Similar argument works for  $|\lambda| \leq \sqrt{a b}$.  If $b'>a'$ the proof is similar. \hfill $\Box$

\section*{Acknowledgments}
 DC was supported by the Polish National Science Center project No. 2018/30/A/ST2/00837. We thank Krzysztof Szczygielski for interesting discussions and his help in preparing plots.

\end{document}